\documentclass{article}
 
\usepackage{lineno}
\usepackage{xspace}
\usepackage{graphics}
\usepackage{amsmath,amsthm,amsfonts}

\title{A linear-time algorithm for the geodesic center of a simple polygon}

\author{
Hee-Kap Ahn\thanks{Department of Computer Science and Engineering, POSTECH, Korea.\texttt{\{heekap, jin9082\}@postech.ac.kr}. Supported by the NRF grant 2011-0030044 (SRC-GAIA) funded by the Korea government (MSIP).}
\and Luis Barba\thanks{School of Computer Science, Carleton University, Ottawa, Canada. \texttt{jit@scs.carleton.ca, jdecaruf@cg.scs.carleton.ca}} $^,$\thanks{D\'epartement d'Informatique, Universit\'e Libre de Bruxelles, Brussels, Belgium. \texttt{lbarbafl@ulb.ac.be}}
 \and Prosenjit Bose\footnotemark[2] 
 \and Jean-Lou de Carufel\footnotemark[2] 
 \and Matias Korman\thanks{National Institute of Informatics (NII), Tokyo, Japan.\texttt{korman@nii.ac.jp}} $^,$\thanks{JST, ERATO, Kawarabayashi Large Graph Project} 
 \and Eunjin Oh\footnotemark[1]
 }

\newtheorem{theorem}{Theorem}[section]
\newtheorem{lemma}[theorem]{Lemma}
\newtheorem{corollary}[theorem]{Corollary}

\newcommand{\F}[2]{\ensuremath{F_{\scriptscriptstyle #1}(#2)}}
\newcommand{\f}[2]{\ensuremath{f_{\scriptscriptstyle #1}(#2)}}
\newcommand{\fn}[2]{\ensuremath{S_{\scriptscriptstyle #1}(#2)}}
\newcommand{\ff}[1]{\ensuremath{f(#1)}}
\newcommand{\cp}{\ensuremath{c_P}}
\newcommand{\m}{\ensuremath{m_{\scriptscriptstyle R}}}

\newcommand{\g}[2]{\ensuremath{|\pi(#1, #2)|}}
\newcommand{\p}[2]{\ensuremath{\pi(#1, #2)}}
\newcommand{\reg}{\ensuremath{R'}}
\newcommand{\tcell}{4-cell\xspace}
\newcommand{\tcells}{4-cells\xspace}

\newcommand{\luis}[2][says]{}
\newcommand{\mati}[2][says]{}

\begin{document}

\maketitle

\begin{abstract}
Given two points in a simple polygon $P$ of $n$ vertices, its geodesic distance is the length of the shortest path that connects them among all paths that stay within $P$. The  geodesic center of $P$ is the unique point in $P$ that minimizes the largest geodesic distance to all other points of $P$. In 1989, Pollack, Sharir and Rote [Disc. \& Comput. Geom. 89] showed an $O(n\log n)$-time algorithm that computes the geodesic center of $P$. Since then, a longstanding question has been whether this running time can be improved (explicitly posed by Mitchell [Handbook of Computational Geometry, 2000]). 
In this paper we affirmatively answer this question and present a linear time algorithm to solve this problem.
\end{abstract}


\section{Introduction}
Let $P$ be a simple polygon with $n$ vertices.
Given two points $x,y$ in $P$, the \emph{geodesic path} $\p{x}{y}$ is the shortest-path contained in $P$ connecting $x$ with $y$. If the straight-line segment connecting $x$ with $y$ is contained in $P$, then $\p{x}{y}$ is a straight-line segment. Otherwise, $\p{x}{y}$ is a polygonal chain whose vertices (other than its endpoints) are  reflex vertices of $P$. We refer the reader to~\cite{m-gspno-00} for more information on geodesic paths refer.

The \emph{geodesic distance} between $x$ and $y$, denoted by $\g{x}{y}$, is the sum of the Euclidean lengths of each segment in $\p{x}{y}$. Throughout this paper, when referring to the distance between two points in $P$, we refer to the geodesic distance between them. Given a point $x\in P$, a (geodesic) \emph{farthest neighbor} of $x$, is a point $\f{P}{x}$ (or simply $\ff{x}$) of $P$ whose geodesic distance to $x$ is maximized. 
To ease the description, we assume that each vertex of $P$ has a unique farthest neighbor. 
We can make this \emph{general position} assumption using simulation of simplicity~\cite{edelsbrunner1990simulation}.

Let $\F{P}{x}$ be the function that, for each $x\in P$, maps to the distance to a farthest neighbor of $x$ (i.e., $\F{P}{x} = \g{x}{\ff{x}}$).
A point $\cp\in P$ that minimizes $\F{P}{x}$ is called the \emph{geodesic center} of $P$. Similarly, a point $s\in P$ that maximizes $\F{P}{x}$ (together with its farthest neighbor) is called a \emph{geodesic diametral pair} and their distance is known as the \emph{geodesic diameter}. Asano and Toussaint~\cite{at-cgcsp-85} showed that the geodesic center is unique (whereas it is easy to see that several geodesic diametral pairs may exist).

In this paper, we show how to compute the geodesic center of $P$ in $O(n)$ time.

\subsection{Previous Work}
Since the early 1980s the problem of computing the geodesic center (and its counterpart, the geodesic diameter) has received a lot of attention from the computational geometry community. Chazelle~\cite{c-tpca-82} gave the first algorithm for computing the geodesic diameter (which runs in $O(n^2)$ time using linear space). Afterwards, Suri~\cite{suri1989computing} reduced it to $O(n\log n)$-time without increasing the space constraints. Finally, Hershberger and Suri~\cite{hershberger1993matrix} presented a fast matrix search technique, one application of which is a linear-time algorithm for computing the diameter.

The first algorithm for computing the geodesic center was given by Asano and Toussaint~\cite{at-cgcsp-85}, and runs in $O(n^4\log n)$-time. In 1989, Pollack, Sharir, and Rote~\cite{pollackComputingCenter} improved it to $O(n\log n)$ time. Since then, it has been an open problem whether the geodesic center can be
computed in linear time (indeed, this problem was explicitly posed by Mitchell~\cite[Chapter 27]{m-gspno-00}).

Several other variations of these two problems have been considered. Indeed, the same problem has been studied under different metrics. Namely, the $L_1$ geodesic distance~\cite{bkow-clgdcsplt-13},  the link distance~\cite{suri-mlpprp-87,k-ealdp-89,dls-aclcsp-92} (where we look for the path with the minimum possible number of bends or {\em links}), or even rectilinear link distance~\cite{ns-crldp-91,ns-oarlcrp-96} (a variation of the link distance in which only isothetic segments are allowed). The diameter and center of a simple polygon for both the $L_1$ and rectilinear link metrics can be computed in linear time (whereas $O(n\log n)$ time is needed for the link distance).

Another natural extension is the computation of the diameter and center in polygonal domains (i.e., polygons with one or more holes). Polynomial time algorithms are known for both the diameter~\cite{bko-gdpd-13} and center~\cite{bko-cgcpd-14}, although the running times are significantly larger (i.e., $O(n^{7.73})$ and $O(n^{12+\varepsilon})$, respectively).

\subsection{Outline}
In order to compute the geodesic center, Pollack {\em et al.}~\cite{pollackComputingCenter} introduce a linear time \emph{chord-oracle}. Given a chord $C$ that splits $P$ into two sub-polygons, the oracle determines which sub-polygon contains $\cp$. Combining this operation with an efficient search on a triangulation of $P$, Pollack {\em et al.} narrow the search of $\cp$ within a triangle (and find the center using optimization techniques). Their approach however, does not allow them to reduce the complexity of the problem in each iteration, and hence it runs in $\Theta(n\log n)$ time. 

The general approach of our algorithm described in Section~\ref{section:Prune and search} is similar: 
partition $P$ into $O(1)$ cells, use an oracle to determine which cell contains $\cp$, and recurse within the cell. 
Our approach differs however in two important aspects that allows us to speed-up the algorithm. 
First, we do not use the chords of a triangulation of $P$ to partition the problem into cells. 
We use instead a cutting of a suitable set of chords.
Secondly, we compute a set $\Phi$ of $O(n)$ functions, each defined in a triangular domain contained in $P$, such that their upper envelope, $\phi(x)$, coincides with $\F{P}{x}$. 
Thus, we can ``ignore'' the polygon $P$ and focus only on finding the minimum of the function $\phi(x)$. 

The search itself uses $\varepsilon$-nets and cutting techniques,
which certify that both the size of the cell containing $\cp$ and the number of functions of $\Phi$ defined in it decrease by a constant fraction 
(and thus leads to an overall linear time algorithm).
This search has however two stopping conditions, (1) reach a subproblem of constant size, 
or (2) find a triangle containing $\cp$.
In the latter case, 
we show that $\phi(x)$ is a convex function when restricted to this triangle. 
Thus, finding its minimum becomes an optimization problem that we solve in Section~\ref{Section:Solving convex optimization poblem} using cuttings in $\mathbb{R}^3$.

The key of this approach lies in the computation of the functions of $\Phi$ and their triangular domains. 
Each function $g(x)$ of $\Phi$ is defined in  a triangular domain $\triangle$ contained in $P$ and is associated to a particular vertex $w$ of $P$. 
Intuitively speaking, $g(x)$ maps points in $\triangle$ to their (geodesic) distance to $w$.
We guarantee that, for each point $x\in P$, there is one function $g$ defined in a triangle containing $x$, such that $g(x) = \F{P}{x}$.
To compute these triangles and their corresponding functions, we proceed as follows.

In Section~\ref{Section:Decomposing the boundary}, we use the matrix search technique introduced by Hershberger and Suri~\cite{hershberger1993matrix} to decompose the boundary of $P$, denoted by $\partial P$, into connected edge disjoint chains.
Each chain is defined by either $(1)$ a consecutive list of vertices that have the same farthest neighbor $v$ (we say that $v$ is \emph{marked} if it has such a chain associated to it), or $(2)$ an edge whose endpoints have different farthest neighbors (such edge is called a \emph{transition edge}).

In Section~\ref{Section: Building hourglasses}, we consider each transition edge $ab$ of $\partial P$ independently and compute its \emph{hourglass}. Intuitively, the hourglass of $ab$, $H_{ab}$, is the region of $P$ between two chains, the edge $ab$ and the chain of $\partial P$ that contains the farthest neighbors of all points in $ab$.
Inspired by a result of Suri~\cite{suri1989computing}, we show that the sum of the complexities of each hourglass defined on a transition edge is $O(n)$. 
In addition, we provide a new technique to compute all these hourglasses in linear time.

In Section~\ref{Section:Computing apexed triangles} we show how to compute the functions in $\Phi$ and their respective triangles.
We distinguish two cases: (1) Inside each hourglass $H_{ab}$ of a transition edge, we use a technique introduced by Aronov et al.~\cite{aronov1993furthest} that uses the shortest-path trees of $a$ and $b$ in $H_{ab}$ to decompose $H_{ab}$ into $O(|H_{ab}|)$ triangles with their respective functions (for more information on shortest-path trees refer to~\cite{guibasShortestPathTree}). 
(2) For each marked vertex $v$ we compute triangles that encode the distance from $v$. Moreover, we guarantee that these triangles cover every point of $P$ whose farthest neighbor is $v$. 
Overall, we compute the $O(n)$ functions of $\Phi$ in linear time.

\section{Hourglasses and Funnels}
In this section, we introduce the main tools that are going to be used by the algorithm. Some of the results presented in this section have been shown before in different papers. 

\subsection{Hourglasses}

Given two points $x$ and $y$ on $\partial P$, let $\partial P(x,y)$ be the polygonal chain that starts at $x$ and follows the boundary of $P$ clockwise until reaching $y$.

For any polygonal chain $C= \partial P(p_0, p_1, \ldots, p_k)$, the \emph{hourglass} of $C$, denoted by $H_C$, is the simple polygon contained in $P$ bounded by $C$, $\p{p_k}{\ff{p_0}}$, $\partial P(\ff{p_0}, \ff{p_k})$ and $\p{\ff{p_k}}{ p_0}$; see Figure~\ref{fig:Transition chains and hourglasses}. 
We call $C$ and $\partial P(\ff{p_0}, \ff{p_k})$ the \emph{top} and \emph{bottom} chains of $H_C$, respectively, while $\p{p_k}{ \ff{p_0}}$ and $\p{\ff{p_k}}{p_0}$ are referred to as the \emph{walls} of $H_C$.

\begin{figure}[tb]
\centering
\includegraphics{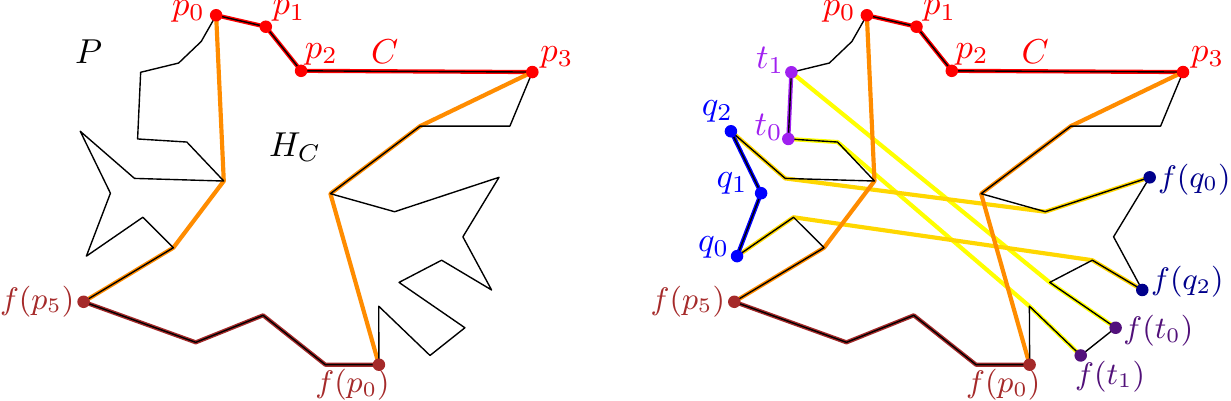}
\caption{\small Given two edge disjoint transition chains, their hourglasses are open and
the bottom chains of their hourglasses are also edge disjoint. 
Moreover, these bottom chains appear in the same cyclic order as the top chains along $\partial P$.}
\label{fig:Transition chains and hourglasses}
\end{figure}

We say that the hourglass $H_C$ is \emph{open} if its walls are vertex disjoint. 
We say $C$ is a \emph{transition chain} if $\ff{p_0} \neq \ff{p_k}$ and neither $\ff{p_0}$ nor $\ff{p_k}$ are interior vertices of $C$. In particular, if an edge $ab$ of $\partial P$ is a transition chain, we say that it is a \emph{transition edge} (see Figure~\ref{fig:Transition chains and hourglasses}).

\begin{lemma}\label{lemma:Transition hourglasses are open}
[Rephrase of Lemma~3.1.3 of~\cite{aronov1993furthest}] 
If $C$ is a transition chain of $\partial P$, then the hourglass $H_C$ is an open hourglass.
\end{lemma}

Note that by Lemma~\ref{lemma:Transition hourglasses are open}, the hourglass of each transition chain is open.
In the remainder of the paper, all the hourglasses considered are defined by a transition chain, i.e., they are open and their top and bottom chains are edge disjoint.

The following lemma is depicted in Figure~\ref{fig:Transition chains and hourglasses} and is a direct consequence of the Ordering Lemma proved by Aronov et al.~\cite[Corollary 2.7.4]{aronov1993furthest}.
\begin{lemma}\label{lemma:Ordering Lemma}
Let $C_1, C_2, C_3$ be three edge disjoint transition chains of $\partial P$ that appear in this order when traversing clockwise the boundary of $P$. Then, the bottom chains of $H_{C_1}, H_{C_2}$ and $H_{C_3}$ are also edge disjoint and appear in this order when traversing clockwise the boundary of~$P$.
\end{lemma}

Let $\gamma$ be a  geodesic path joining two points on the boundary of $P$.
We say that $\gamma$ \emph{separates} two points $x_1$ and $x_2$ of $\partial P$ if the points of $X=\{x_1, x_2\}$ and the endpoints of $\gamma$ alternate along the boundary of $P$ ($x_1$ and $x_2$ could coincide with the endpoints of $\gamma$ in degenerate cases). We say that a geodesic path $\gamma$ \emph{separates} an hourglass $H$ if it separates the points of its top chain from those of its  bottom chain.

\begin{lemma}\label{lemma:Split paths}
Let $C_1, \ldots, C_r$ be edge disjoint transition chains of $\partial P$. Then, there is a set of $t = O(1)$ geodesic paths $\gamma_1, \ldots, \gamma_t$ with endpoints on $\partial P$ such that for each $1\leq i\leq r$ there exists $1\leq j\leq t$ such that $\gamma_j$ separates $H_{C_i}$.
Moreover, this set can be computed in $O(n)$ time.
\end{lemma}
\begin{proof}
Aronov et al. showed that there exist four vertices $v_1, \ldots, v_4$ of $P$ and geodesic paths $\p{v_1}{v_2}, \p{v_2}{v_3}, \p{v_3}{v_4}$ such that for any point $x\in \partial P$, one of these paths separates $x$ from $\ff{x}$~\cite[Lemma 2.7.6]{aronov1993furthest}. Moreover, they show how to compute this set in $O(n)$ time.

Let $\Gamma= \{\p{v_i}{v_j} : 1\leq i < j\leq 4\}$ and note that $v_1, \ldots, v_4$ split the boundary of $P$ into at most four connected components.
If a chain $C_i$ is completely contained in one of these components, then one path of $\Gamma$ separates the top and bottom chain of $H_{C_i}$. Otherwise, some vertex $v_j$ is an interior vertex of $C_i$. However, because the chains $C_1, \ldots, C_r$ are edge disjoint, there are at most four chains in this situation. 
For each chain $C_i$ containing a vertex $v_j$, we add the geodesic path connecting the endpoints of $C_i$ to $\Gamma$.
Therefore, $\Gamma$ consists of $O(1)$ geodesic paths and each hourglass $H_{C_i}$ has its top and bottom chain separated by some path of $\Gamma$.
Since only $O(1)$ additional paths are computed, this can be done in linear time.
\end{proof}

A \emph{chord} of $P$ is an edge joining two non-adjacent vertices $a$ and $b$ of $P$ such that $ab\subseteq P$. Therefore, a chord splits $P$ into two sub-polygons.

\begin{lemma}\label{lemma:Edges appear a constant number of times}
[Rephrase of Lemma 3.4.3 of~\cite{aronov1993furthest}]
Let $C_1, \ldots, C_r$ be a set of edge disjoint transition chains of $\partial P$ that appear in this order when traversing clockwise the boundary of $P$. Then each chord of $P$ appears in $O(1)$ hourglasses among $H_{C_1}, \ldots, H_{C_r}$.
\end{lemma}
\begin{proof}
Note that chords can only appear on walls of hourglasses. Because  hourglasses are open, any chord must be an edge on exactly one wall of each of these hourglasses. Assume, for the sake of contradiction, that there exists two points $s,t\in P$ whose chord $st$ is in three hourglasses $H_{C_i}, H_{C_j}$ and $H_{C_k}$ (for some $1\leq i < j < k\leq r$) such that $s$ visited before $t$ when going from the top chain to the bottom one along the walls of the three hourglasses.
Let $s_i$ and $t_i$ be the points in the in the top and bottom chains of $H_{C_i}$, respectively, such that $\p{s_i}{t_i}$ is the wall of $H_{C_i}$ that contains $st$ (analogously, we define $s_k$ and $t_k$)

Because  $C_j$ lies in between $C_i$ and $C_k$, Lemma~\ref{lemma:Ordering Lemma} implies that the bottom chain of $C_j$ appears between the bottom chains of $C_i$ and $C_k$. Therefore, $C_j$ lies between $s_i$ and $s_k$ and the bottom chain of $H_{C_j}$ lies between $t_i$ and $t_k$. 
That is, for each $x\in C_j$ and each $y$ in the bottom chain of $H_{C_j}$, the geodesic path $\p{x}{y}$ is ``sandwiched'' by the paths $\p{s_i}{t_i}$ and $\p{s_k}{t_k}$.
In particular, $\p{x}{y}$ contains $st$ for each pair of points in the top and bottom chain of $H_{C_j}$.
However, this implies that the hourglass $H_{C_j}$ is not open---a contradiction that comes from assuming that $st$ lies in the  wall of three open hourglasses, when this wall is traversed from the top chain to the bottom chain. 
Analogous arguments can be used to bound the total number of walls that contain the edge $st$ (when traversed in any direction) to $O(1)$.
\end{proof}


\begin{lemma}\label{lemma:Suri's lemma}
Let $x, u, y, v$ be four vertices of $P$ that appear in this cyclic order in a clockwise traversal of $\partial P$.
Given the shortest-path trees $T_x$ and $T_y$ of $x$ and $y$ in $P$, respectively, such that $T_x$ and $T_y$ can answer lowest common ancestor (LCA) queries in $O(1)$ time, 
we can compute the path $\p{u}{v}$ in $O(|\p{u}{v}|)$ time. 
Moreover, all edges of $\p{u}{v}$, except perhaps one, belong to $T_x\cup T_y$.
\end{lemma}
\begin{proof}
Let $X$ (\emph{resp.} $Y$) be the set containing the LCA in $T_x$ (\emph{resp.} $T_y$) of $u,y$, and of $v,y$ (\emph{resp.} $u,x$ and $x,y$). Note that the points of $X\cup Y$ lie on the path $\p{x}{y}$ and can be computed in $O(1)$ time by hypothesis. Moreover, using LCA queries, we can decide their order along the path $\p{x}{y}$ when traversing it from $x$ to $y$. (Both $X$ and $Y$ could consist of a single vertex in some degenerate situations). Two cases arise: 

\textbf{Case 1.} If there is a vertex  $x^*\in X$ lying after a vertex $y^*\in Y$ along $\p{x}{y}$, 
then the path $\p{u}{v}$ contains the path $\p{y^*}{x^*}$. 
In this case, the path $\p{u}{v}$ is the concatenation of the paths $\p{u}{y^*}$, $\p{y^*}{x^*}$, and $\p{x^*}{v}$ and that the three paths are contained in $T_x \cup T_y$.
Moreover, $\p{u}{v}$ can be computed in time proportional to its length by traversing along the corresponding tree; see Figure~\ref{fig:Output Sensitive Algorithm} (top).

\textbf{Case 2.} In this case the vertices of $X$ appear before the vertices of $Y$ along $\p{x}{y}$.
Let $x'$ (\emph{resp.} $y'$) be the vertex of $X$ (\emph{resp.} $Y$) closest to $x$ (\emph{resp.} $y$). 

Let $u'$ be the last vertex of $\p{u}{x}$ that is also in $\p{u}{y}$.
Note that $u'$ can be constructed by walking from $u'$ towards $x$ $y$ until the path towards $y$ diverges. 
Thus, $u'$ can be computed in $O(|\p{u}{u'}|)$ time. 
Define $v'$ analogously and compute it in $O(|\p{v}{v'}|)$ time.

Let $P'$ be the polygon bounded by the geodesic paths $\p{x'}{u'}, \p{u'}{y'}, \p{y'}{v'}$ and  $\p{v'}{x'}$.
Because the vertices of $X$ appear before those of $Y$ along $\p{x}{y}$, $P'$ is a simple polygon; see Figure~\ref{fig:Output Sensitive Algorithm} (bottom).

In this case the path $\p{u}{y}$ is the union of $\p{u}{u'}, \p{u'}{v'}$ and $\p{v'}{v}$.
Because $\p{u}{u'}$ and  $\p{v'}{v}$ can be computed in time proportional to their length, it suffices to compute $\p{u'}{v'}$ in $O(|$\p{u'}{v'}$|)$ time.

Note that $P'$ is a simple polygon with only four convex vertices $x',u', y'$ and $v'$, which are connected by chains of reflex vertices.
Thus, the shortest path from $x'$ to $y'$ can have at most one diagonal edge connecting distinct reflex chains of $P'$. Since the rest of the points in $\p{u'}{v'}$ lie on the boundary of $P'$ and from the fact that each edge of $P'$ is an edge of $T_x\cup T_y$, we conclude all edges of $\p{u}{v}$, except perhaps one, belong to $T_x\cup T_y$.

We want to find the common tangent between the reflex paths $\p{u'}{x'}$ and $\p{v'}{y'}$, or the common tangent of $\p{u'}{y'}$ and $\p{v'}{x'}$ as one of them belongs to the shortest path $\p{u'}{v'}$.
Assume that the desired tangent lies between the paths $\p{u'}{x'}$ and $\p{v'}{y'}$. 
Since these paths consist only of reflex vertices, the problem can be reduced to finding the common tangent of two convex polygons. 
By slightly modifying the linear time algorithm to compute this tangents, we can make it run in $O(|\p{u'}{v'}|)$ time.

Since we do not know if the tangent lies between the paths $\p{u'}{x'}$ and $\p{v'}{y'}$, we process the chains $\p{u'}{y'}$ and $\p{v'}{x'}$ in parallel and stop when finding the desired tangent. Consequently, we can compute the path $\p{u}{v}$ in time proportional to its length. 
\end{proof}

\begin{figure}[tb]
\centering
\includegraphics{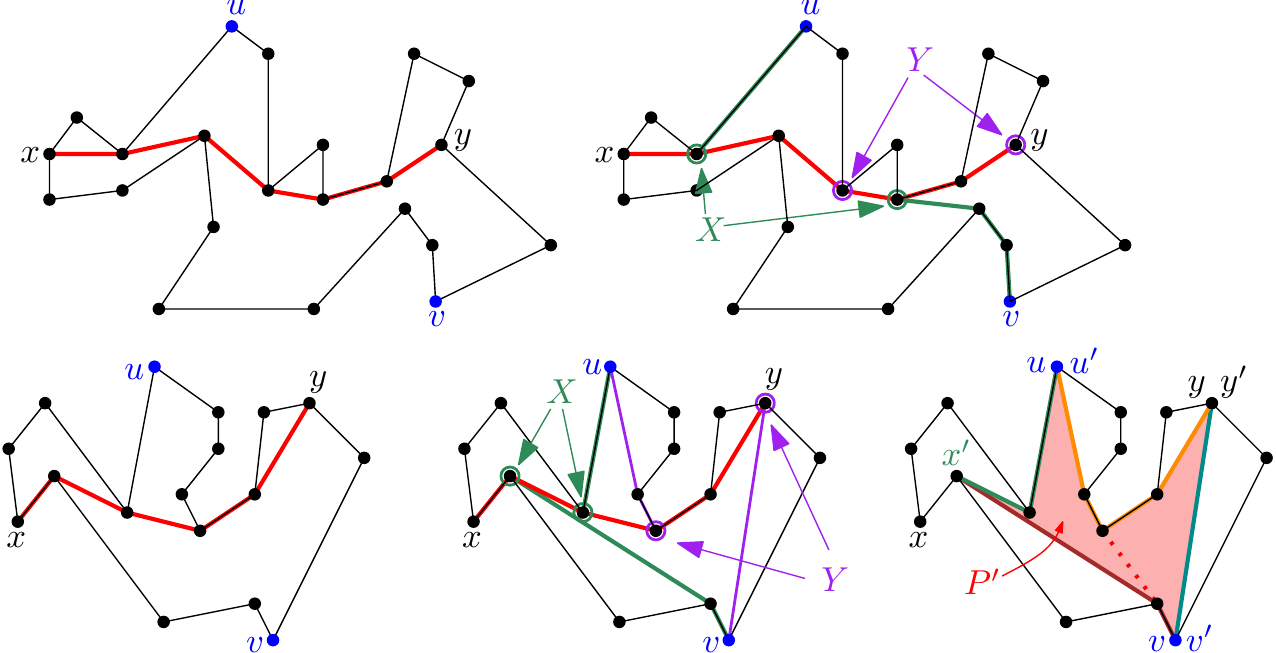}

\caption{\small (top) Case 1 of the proof of Lemma~\ref{lemma:Suri's lemma} where the path $\p{u}{v}$ contains a portion of the path $\p{x}{y}$.
(bottom) Case 2 of the proof of Lemma~\ref{lemma:Suri's lemma} where the path $\p{u}{v}$ has exactly one edge being the tangent of the paths $\p{u'}{y'}$ and $\p{v'}{x'}$.}
\label{fig:Output Sensitive Algorithm}
\end{figure}

\begin{lemma}\label{lemma:Bounding complexity of hourglasses}
Let $P$ be a simple polygon with $n$ vertices.
Given $k$ disjoint transition chains $C_1, \ldots, C_k$  of $\partial P$, it holds that  $$\sum_{i=1}^k |H_{C_i}| = O(n).$$
\end{lemma}
\begin{proof}

Because the given transition chains are disjoint, Lemma~\ref{lemma:Ordering Lemma} implies that the bottom chains of their respective hourglasses are also disjoint. Therefore, the sum of the complexities of all the top and bottom chains of these hourglasses is $O(n)$. 
To bound the complexity of their walls we use
Lemma~\ref{lemma:Edges appear a constant number of times}. Since  no chord is used more than a constant number of times, it suffices to show that the total number of chords used by all these hourglasses is $O(n)$.

To prove this, we use Lemma~\ref{lemma:Split paths} to construct $O(1)$ \emph{split chains} $\gamma_1, \ldots, \gamma_t$ such that for each $1\leq i\leq k$, there is a split chain $\gamma_j$ that separates the top and bottom chains of $H_{C_i}$.
For each $1\leq j\leq t$, let $$\mathcal H^j = \{H_{C_i} : \text{the top and bottom chain of $H_{C_i}$ are separated by }\gamma_j\}.$$
Since the complexity of the shortest-path trees of the endpoints of $\gamma_j$ is $O(n)$~\cite{guibasShortestPathTree},
and from the fact that the chains $C_1, \ldots, C_k$ are disjoint,  Lemma~\ref{lemma:Suri's lemma} implies that
the total number of edges in all the hourglasses of $\mathcal H^j$ is $O(n)$. Moreover, because each of these edges appears in $O(1)$ hourglasses among $C_1, \ldots, C_k$, we conclude that 
$$\sum_{H \in \mathcal H^j } |H| = O(n).$$
Since we have only $O(1)$ split chains, our result follows.
\end{proof}

\subsection{Funnels}

Let $C = (p_0, \ldots, p_k)$ be a chain of $\partial P$ and let $v$ be a vertex of $P$ not in $C$.
The \emph{funnel} of $v$ to $C$, denoted by $\fn{v}{C}$, is the simple polygon bounded by $C$, $\p{p_k}{v}$ and $\p{v}{p_0}$; see Figure~\ref{fig:Marked vertices decomposition} $(a)$. 
Note that the paths $\p{v}{p_k}$ and $\p{v}{p_0}$ may coincide for a while before splitting into disjoint chains. 
See Lee and Preparata~\cite{lee1984euclidean} or Guibas et al.~\cite{guibasShortestPathTree} for more details on funnels.

A subset $R\subset P$ is \emph{geodesically convex} if for every $x,y\in R$, the path $\p{x}{y}$ is contained in $R$.
This funnel $\fn{v}{C}$ is also known as the geodesic convex hull of $C$ and $v$, i.e., the minimum geodesically convex set that contains $v$ and $C$.

Given two points $x,y\in P$, the (geodesic) \emph{bisector} of $x$ and $y$ is the set of points contained in $P$ that are equidistant from $x$ and $y$. This bisector is a curve, contained in $P$, that consists of circular arcs and hyperbolic arcs. Moreover, this curve intersects $\partial P$ only at its endpoints~\cite[Lemma 3.22]{aronov1989geodesic}.

The (farthest) \emph{Voronoi region} of a vertex $v$ of $P$ is the set of points $R(v) = \{x\in P : \F{P}{x} = \g{x}{v}\}$ (including boundary points).

\begin{lemma}\label{lemma:Funnel contains Voronoi region}
Let $v$ be a vertex of $P$ and let $C$ be a transition chain such $R(v)\cap \partial P \subseteq C$ and $v\not\in C$.
Then, $R(v)$ is contained in the funnel $\fn{v}{C}$
\end{lemma}
\begin{proof}
Let $a$ and $b$ be the endpoints of $C$ such that $a,b, \ff{a}$ and $\ff{b}$ appear in this order in a clockwise traversal of $\partial P$.
Because $R(v)\cap \partial P\subset C$, we know that $v$ lies between $\ff{a}$ and $\ff{b}$.

Let $\alpha$ (\emph{resp.} $\beta$) be the bisector of $v$ and $\ff{a}$ (\emph{resp.} $\ff{b}$).
Let $h_a$ (\emph{resp.} $h_b$) be the set of points of $P$ that are farther from $v$ than from $\ff{a}$ (\emph{resp.} $\ff{b}$).
Note that $\alpha$ is the boundary of $h_a$ while $\beta$ bounds $h_b$.

By definition, we know that $R(v)\subseteq h_a\cap h_b$. Therefore, it suffices to show that $h_a\cap h_b\subset \fn{v}{C}$.
Assume for a contradiction that there is a point of $h_a\cap h_b$ lying outside of $\fn{v}{C}$. 
By continuity of the geodesic distance, the boundaries of $h_a\cap h_b$ and $\fn{v}{C}$ must intersect.
Because $a\notin h_a$ and $b\notin h_b$, both bisectors $\alpha$ and $\beta$ must have an endpoint on the edge $ab$.
Since the boundaries of $h_a\cap h_b$ and $\fn{v}{C}$ intersect, we infer that $\beta \cap \p{v}{b}\neq \emptyset$ or $\alpha \cap \p{v}{a}\neq \emptyset$.
Without loss of generality, assume that there is a point $w\in \beta \cap \p{v}{b}$, the case where $w$ lies in $\alpha \cap \p{v}{a}$ is analogous. 

Since $w\in \beta$, we know that $\g{w}{v} = \g{w}{ \ff{b}}$. By the triangle inequality and since $w$ cannot be a vertex of $P$ as $w$ intersects $\partial P$ only at its endpoints, we get that
$$\g{b}{\ff{b}} < \g{b}{w} + \g{w}{\ff{b}} = \g{b}{w} + \g{w}{v} = \g{b}{v}.$$
Which implies that $b$ is farther from $v$ than from $\ff{b}$---a contradiction that comes from assuming that $h_a\cap h_b$ is not contained in $\fn{v}{C}$.
\end{proof}

\section{Decomposing the boundary}\label{Section:Decomposing the boundary}
In this section, we decompose the boundary of $P$ into consecutive vertices that share the same farthest neighbor and edges of $P$ whose endpoints have distinct farthest neighbors.

Using a result from Hershberger and Suri~\cite{hershberger1993matrix}, in $O(n)$ time we can compute the farthest neighbor of each vertex of $P$.
Recall that the farthest neighbor of each vertex of $P$ is always a convex vertex of $P$~\cite{at-cgcsp-85} and is unique by our general position assumption. 

We mark the vertices of $P$ that are farthest neighbors of at least one vertex of $P$. Let $M$ denote the set of marked vertices of $P$ (clearly this set can be computed in $O(n)$ time after applying the result of Hershberger and Suri).
In other words, $M$ contains all vertices of $P$ whose Voronoi region contains at least one vertex of $P$.

Given a vertex $v$ of $P$, the vertices of $P$ whose farthest neighbor is $v$ appear contiguously along $\partial P$~\cite{aronov1993furthest}. Therefore, after computing all these farthest neighbors, we effectively split the boundary into subchains, each associated with a different vertex of $M$; see Figure~\ref{fig:Marked vertices decomposition} $(b)$.

\begin{figure}[tb]
\centering
\includegraphics{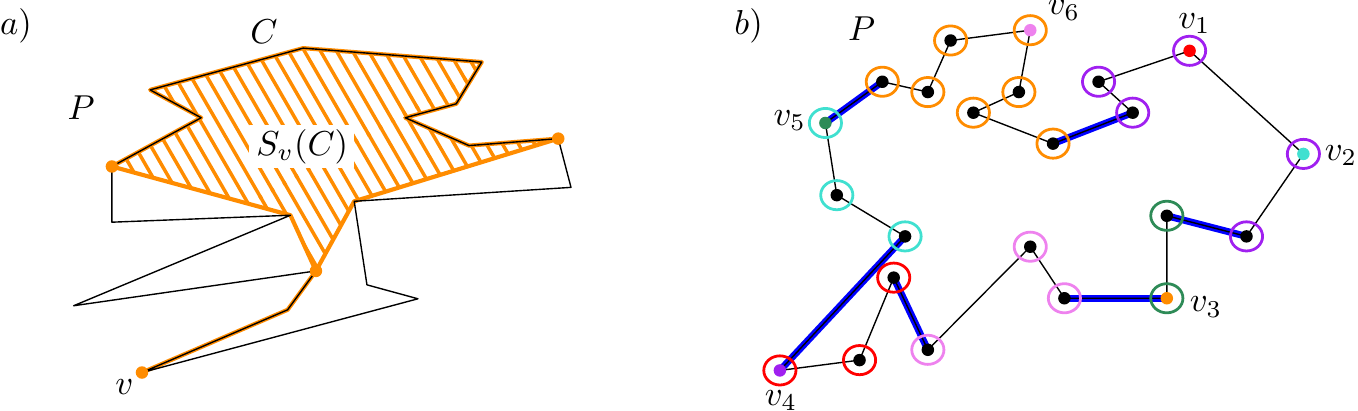}

\caption{\small 
$a)$ The funnel $\fn{v}{C}$ of a vertex $v$ and a chain $C$ contained in $\partial P$ are depicted.
$b)$ Each vertex of the boundary of $P$ is assigned with a farthest neighbor which is then marked. 
The boundary is then decomposed into vertex disjoint chains, each associated with a marked vertex, joined by transition edges (blue) whose endpoints have different farthest neighbors.}
\label{fig:Marked vertices decomposition}
\end{figure}

Let $a$ and $b$ be the endpoints of a transition edge of $\partial P$ such that $a$ appears before $b$ in the clockwise order along $\partial P$. Because $ab$ is a transition edge, we know that $\ff{a}\neq \ff{b}$.
Recall that we have computed $\ff{a}$ and $\ff{b}$ in the previous step and note that $\ff{a}$ appears also before $\ff{b}$ along this clockwise order. 
For every vertex $v$ that lies between $\ff{a}$ and $\ff{b}$ in the bottom chain of $H_{ab}$, we know that there cannot be a vertex $u$ of $P$ such that $\ff{u} = v$.
As proved by Aronov et al.~\cite[Corollary 2.7.4]{aronov1993furthest}, 
if there is a point $x$ on $\partial P$ whose farthest neighbor is $v$, then $x$ must lie on the open segment $(a,b)$. 
In other words, the Voronoi region $R(v)$ restricted to $\partial P$ is contained in~$(a,b)$.

\section{Building hourglasses}\label{Section: Building hourglasses}

Let $E$ be the set of transition edges of $\partial P$.
Given a transition edge $ab\in E$, we say that $H_{ab}$ is a \emph{transition hourglass}.
In order to construct the triangle cover of $P$, we construct the transition hourglass of each transition edge of $E$. By Lemma~\ref{lemma:Bounding complexity of hourglasses}, we know that $\sum_{ab\in E} |H_{ab}| = O(n)$. Therefore, our aim is to compute the cover in time proportional to the size of $H_{ab}$. 

By Lemma~\ref{lemma:Split paths} we can compute a set of $O(1)$ separating paths such that for each transition edge $ab$, the transition hourglass $H_{ab}$ is separated by one (or more) paths in this set. For each endpoint of the $O(1)$ separating paths we compute its shortest-path tree~\cite{guibasShortestPathTree}. In addition, we preprocess these trees in linear time to support LCA queries~\cite{harel1984fast}. Both computations need linear time per endpoint and use $O(n)$ space. Since we do this process for a constant number of endpoints, overall this preprocessing takes $O(n)$ time. 

Let $\gamma$ be a separating path whose endpoints are $x$ and $y$. Note that $\gamma$ separates the boundary of $P$ into two chains $S$ and $S'$ such that $S\cup S' = \partial P$.
Let $\mathcal H(\gamma)$ be the set of each transition hourglass separated by $\gamma$ whose transition edge is contained in $S$ (whenever an hourglass is separated by more than one path, we pick one arbitrarily). Note that we can classify all transition hourglasses into the sets $\mathcal H(\gamma)$ in $O(n)$ time (since $O(1)$ separating paths are considered).

We claim that we can compute all transition hourglass of $\mathcal H(\gamma)$ in $O(n)$ time.
By construction, the wall of each of these hourglasses consists of a (geodesic) path that connects a point in $S$ with a point in $S'$. Let $u\in S$ and $v\in S'$ be two vertices such that $\p{u}{v}$ is the wall of a hourglass in $\mathcal H(\gamma)$.
Because LCA queries can be answered in $O(1)$ time~\cite{harel1984fast}, 
Lemma~\ref{lemma:Suri's lemma} allows us to compute this path in $O(|\p{u}{v}|)$ time. 
Therefore, we can compute all hourglasses of $\mathcal H(\gamma)$ in $O(\sum_{H\in \mathcal H(\gamma)} |H| + n) = O(n)$ time by Lemma~\ref{lemma:Bounding complexity of hourglasses}. 
Because only $O(1)$ separating paths are considered, we obtain the following result.

\begin{lemma}\label{lemma: Hourglass partition}
We can construct the transition hourglass of all transition edges of $P$ in $O(n)$ time.
\end{lemma}

\section{Covering the polygon with apexed triangles}\label{Section:Computing apexed triangles}
An \emph{apexed triangle} $\triangle = (a,b,c)$ with \emph{apex} $a$ is a triangle contained in $P$ with an associated distance function $g_\triangle(x)$, called the \emph{apex function} of $\triangle$, such that (1) $a$ is a vertex of $P$, (2) $b,c \in\partial P$, and (3) there is a  vertex $w$ of  $P$, called the \emph{definer} of $\triangle$, such that
$$g_\triangle(x) = \left\{ \begin{array}{lll}
-\infty&&\text{if $x\notin \triangle$}\\
|xa| + \g{a}{w} = \g{x}{w} && \text{if $x\in \triangle$}
\end{array}\right.$$

In this section, we show how to find a set of $O(n)$ apexed triangles of $P$ such that the upper envelope of their apex functions coincides with $\F{P}{x}$.
To this end, we first decompose the transition hourglasses into apexed triangles that encode all the geodesic distance information inside them. For each marked vertex $v\in M$ we construct a funnel that contains the Voronoi region of $v$.  We then decompose this funnel into apexed triangles that encode the distance from $v$.

\subsection{Inside the transition hourglass}
Let $ab$ be a transition edge of $P$  such that $b$ is the clockwise neighbor of $a$ along $\partial P$.
Let $B_{ab}$ denote the bottom chain of $H_{ab}$ after removing its endpoints.
As noticed above, a point on $\partial P$ can be farthest from a vertex in $B_{ab}$ only if it lies in the open segment $ab$.
That is, if $v$ is a vertex of $B_{ab}$ such that $R(v)\neq \emptyset$, then $R(v)\cap \partial P \subset ab$.

In fact, not only this Voronoi region is inside $H_{ab}$ when restricted to the boundary of $P$, but also $R(v)\subset H_{ab}$. 
 The next result follows trivially from Lemma~\ref{lemma:Funnel contains Voronoi region}.

\begin{corollary}\label{lemma:Cell contained in geodesic triangle}
Let $v$ be a vertex of $B_{ab}$. If $R(v)\neq \emptyset$, then $R(v) \subset H_{ab}$.
\end{corollary}

Our objective is to compute $O(|H_{ab}|)$ apexed triangles that cover $H_{ab}$, each with its distance function, such that the upper envelope of these apex functions coincides with $\F{P}{x}$ restricted to $H_{ab}$ where it ``matters''.

The same approach was already used by Pollack et al. in~\cite[Section 3]{pollackComputingCenter}. 
Given a segment contained in the interior of $P$, they show 
how to compute a linear number of apexed triangles such that $\F{P}{x}$ coincides with the upper envelope of the corresponding apex functions in the given segment.

While the construction we follow is analogous, we use it in the transition hourglass $H_{ab}$ instead of the full polygon $P$. 
Therefore, we have to specify what is the relation between the upper envelope of the computed functions and $\F{P}{x}$. 
We will show that the upper envelope of the apex functions computed in $H_{ab}$ coincides with $\F{P}{x}$ inside the Voronoi region $R(v)$ of every vertex $v\in B_{ab}$.

Let $T_a$ and $T_b$ be the shortest-path trees in $H_{ab}$ from $a$ and $b$, respectively. Assume that $T_a$ and $T_b$ are rooted at $a$ and $b$, respectively.
We can compute these trees in $O(|H_{ab}|)$ time~\cite{guibasShortestPathTree}. 
For each vertex $v$ between $\ff{a}$ and $\ff{b}$, let $v_a$ and $v_b$ be the neighbors of $v$ in the paths $\p{v}{a}$ and $\p{v}{ b}$, respectively.
We say that a vertex $v$ is \emph{visible} from $ab$ if $v_a\neq v_b$.
Note that if a vertex is visible, then the extension of these segments must intersect the top segment $ab$. 
Therefore, for each visible vertex $v$, we obtain a triangle $\triangle_v$ as shown in Figure~\ref{fig:Hourglass Cover}.

\begin{figure}[tb]
\centering
\includegraphics{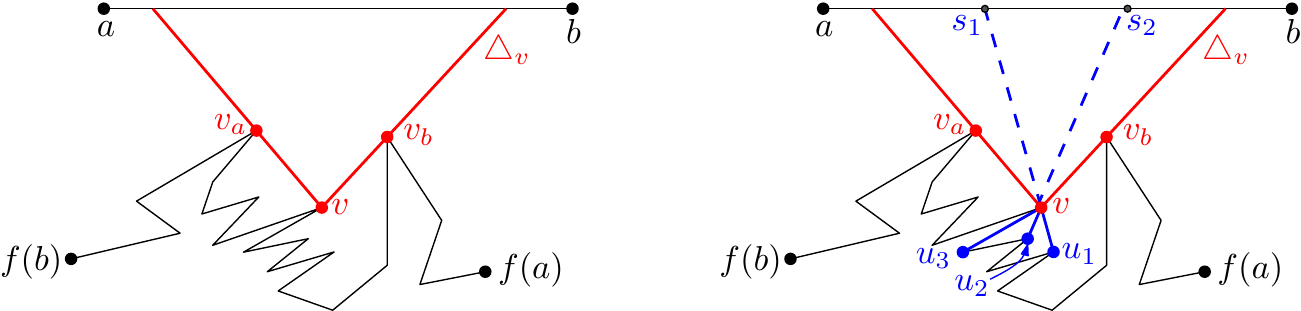}

\caption{\small (left) A vertex $v$ visible from the segment $ab$ lying on the bottom chain of $H_{ab}$, and the triangle $\triangle_v$ which contains the portion of $ab$ visible from $v$. (right) The children $u_1$ and $u_2$ of $v$ are visible from $ab$ while $w_3$ is not. The triangle  $\triangle_v$ is split into apexed triangles by the rays going from $u_1$ and $u_2$ to $v$. }
\label{fig:Hourglass Cover}
\end{figure}

We further split $\triangle_v$ into a series of triangles with apex at $v$ as follows: 
Let $u$ be a child of $v$ in either $T_a$ or $T_b$. As noted by Pollack et al., $v$ can be of three types, either (1) $u$ is not visible from $ab$ (and is hence a child of $v$ in both $T_a$ and $T_b$); or (2) $u$ is visible from $ab$, is a child of $v$ only in $T_b$, and $v_b v u$ is a left turn; or (3) $u$ is visible from $ab$, is a child of $v$ only in $T_a$, and $v_a v u$ is a right turn.

Let $u_1, \ldots, u_{k-1}$ be the children of $v$ of type $(2)$ sorted in clockwise order around $v$.
Let $c(v)$ be the maximum distance from $v$ to any invisible vertex in the subtrees of $T_a$ and $T_b$ rooted at $v$; if no such vertex exists, then $c(v) = 0$. 
Define a function $d_l(v)$ on each vertex $v$ of $H_{ab}$ in a recursive fashion as follows:
If $v$ is invisible from $ab$, then $d_l(v) = c(v)$. 
Otherwise, let $d_l(v)$ be the maximum of $c(v)$ and $\max\{d_l(u_i) + |u_iv| : u_i$ is a child of $v$ of type $(2)\}$.
Similarly we define a symmetric function $d_r(v)$ using the children of type (3) of $v$.

For each $1\leq i\leq k-1$, extend the segment $u_iv$ past $v$ until it intersects $ab$ at a point $s_i$. Let $s_0$ and $s_k$ be the intersections of the extensions of $vv_a$ and $vv_b$ with the segment $ab$.
We define then $k$ triangles contained in $\triangle_v$ as follows. 
For each $0\leq i\leq k-1$, consider the triangle $\triangle(s_i, v, s_{i+1})$ whose associated apexed (left) function is 
$$f_i(x) = \left\{ \begin{array}{lll}
|xv| + \max_{j>i}\{c(v), |vu_j| + d_l(u_j)\} && \text{if $x\in \triangle(s_i, v, s_{i+1})$}\\
-\infty&&\text{otherwise}
\end{array}\right.$$
In a symmetric manner, we define a set of apexed triangles induced by the type (3) children of~$v$ and their respective apexed (right) functions.

Let $g_1, \ldots, g_r$ and $\triangle_1, \ldots, \triangle_r$ respectively be an enumeration of all the generated apex functions and triangles such that $g_i$ is defined in the triangle $\triangle_i$. Because each function is determined uniquely by a pair of adjacent vertices in $T_a$ or in $T_b$, and since these trees have $O(|H_{ab}|)$ vertices, we conclude that $r = O(|H_{ab}|)$. 

Note that for each $1\leq i\leq r$, the triangle $\triangle_i$ has two vertices on the segment $ab$ and a third vertex, say $a_i$, called its \emph{apex} such that for each $x\in \triangle_i$, $g_i(x) = \g{x}{w_i}$ for some vertex $w_i$ of $H_{ab}$. We refer to $w_i$ as the \emph{definer} of $\triangle_i$. Intuitively, $\triangle_i$ defines a portion of the geodesic distance function from $w_i$ in a constant complexity region. 

\begin{lemma}\label{lemma:Triangles inside hourglasses}
Given a transition edge $ab$ of $P$, we can compute a set $\mathcal A_{ab}$ of $O(|H_{ab}|)$ apexed triangles in $O(|H_{ab}|)$ time with the property that for any point $p\in P$ such that $\ff{p}\in B_{ab}$,
there is an apexed triangle $\triangle\in \mathcal A_{ab}$ with apex function $g$ and definer equal to $\ff{p}$ such that
\begin{enumerate}
\item $p\in \triangle$ and 
\item $g(p) = \F{P}{p}$.
\end{enumerate}
\end{lemma}
\begin{proof}
Because $p\in R(\ff{p})$, Lemma~\ref{lemma:Cell contained in geodesic triangle} implies that $p\in H_{ab}$. 
Consider the path $\p{p}{\ff{p}}$ and let $v$ be the neighbor of $p$ along this path. 
By construction of $\mathcal A_{ab}$, there is a triangle $\triangle\in \mathcal A_{ab}$ apexed at $v$ with definer $w$ that contains $p$. The apex function $g(x)$ of $\triangle$ encodes the geodesic distance from $x$ to $w$. 
Because $\F{P}{x}$ is the upper envelope of all the geodesic functions, we know that $g(p) \leq \F{P}{p}$.

To prove the other inequality, note that if $v = \ff{p}$, then trivially $g(p) = |pv| + \g{v}{w} \geq |pv| = \g{p}{\ff{p}} = \F{P}{p}$. 
Otherwise, let $z$ be the next vertex after $v$ in the path $\p{p}{\ff{p}}$. Three cases arise:

($a$) If $z$ is invisible from $ab$, then so is $\ff{p}$ and hence, 
$$\g{p}{ \ff{p}} = |pv| + \g{v}{ \ff{p}} \leq |pv| + c(v) \leq g(p).$$

($b$) If $z$ is a child of type (2), then $z$ plays the role of some child $u_j$ of $v$ in the notation used during the construction.
In this case:
$$\g{p}{\ff{p}} = |pv| + |v z| + \g{z}{\ff{p}} \leq |pv| + |v u_j| + d_l(u_j) \leq g(p).$$

($c$) If $z$ is a child of type (3), then analogous arguments hold using the (right) distance~$d_r$.

Therefore, regardless of the case $\F{P}{p} = \g{p}{ \ff{p}} \leq g(p)$.

To bound the running time, note that the recursive functions $d_l, d_r$ and $c$ can be computed in $O(|T_a| + |T_b|)$ time. Then, for each vertex visible from $ab$, we can process it in time proportional to its degree in $T_a$ and $T_b$.
Because the sum of the degrees of all vertices in $T_a$ and $T_b$ is $O(|T_a| + |T_b|)$ and from the fact that both $|T_a|$ and $|T_b|$ are $O(|H_{ab}|)$, we conclude that the total running time to construct $\mathcal A_{ab}$ is $O(|H_{ab}|)$.
\end{proof}

In other words, Lemma~\ref{lemma:Triangles inside hourglasses} says that no information on farthest neighbors is lost if we only consider the functions in $\mathcal A_{ab}$ within $H_{ab}$.
In the next section we use a similar approach to construct a set of apexed triangles (and their corresponding apex functions), so as to encode the distance from the vertices of $M$.

\subsection{Inside the funnels of marked vertices}
Recall that for each marked vertex $v\in M$, we know at least of one vertex on $\partial P$ such that $v$ is its farthest neighbor.
For any marked vertex $v$, let $u_1, \ldots, u_{k-1}$ be the vertices of $P$ such that $v = \ff{u_i}$ and assume that they appear in this order when traversing $\partial P$ clockwise. Let $u_0$ and $u_k$ be the neighbors of $u_1$ and $u_{k-1}$ other than $u_2$ and $u_{k-2}$, respectively. Note that both $u_0 u_1$ and $u_{k-1}u_k$ are transition edges of $P$. Thus, we can assume that their transition hourglasses have been computed.

Let $C_v = (u_0, \ldots, u_k)$ and consider the funnel $\fn{v}{C_v}$.
We call $C_v$ the \emph{main chain} of $\fn{v}{C_v}$ while $\p{u_k}{ v}$ and $\p{v}{ u_0}$ are referred to as the \emph{walls} of the funnel.  
Because $v = \ff{u_1} = \ff{u_{k-1}}$, we know that $v$ is a vertex of both $H_{u_0 u_1}$ and  $H_{u_{k-1}u_k}$. 
By definition, we have $\p{v}{ u_0}\subset H_{u_0u_1}$ and $\p{v}{u_k}\subset H_{u_{k-1}u_k}$. Thus, we can explicitly compute both paths $\p{v}{ u_0}$ and $\p{v}{u_k}$ in $O( |H_{u_0 u_1}| + |H_{u_{k-1}u_k}|)$ time.
So, overall, the funnel $\fn{v}{C_v}$ can be constructed in $O(k + |H_{u_0 u_1}| + |H_{u_{k-1}u_k}|)$ time. 
Recall that, by Lemma~\ref{lemma:Bounding complexity of hourglasses}, the total sum of the complexities of the transition hourglasses is $O(n)$. In particular, we can bound the total time needed to construct the funnels of all marked vertices by $O(n)$. 

Since the complexity of the walls of these funnels is bounded by the complexity of the transition hourglasses used to compute them, we get that $$\sum_{v\in M} |\fn{v}{C_v}|  = O\left(n + \sum_{ab\in E} |H_{ab}|\right) = O(n).$$

\begin{lemma}\label{lemma:Farthest points from marked are in funnel}
Let $x$ be a point in $P$. If $v = \ff{x}$ is a vertex of $M$, then $x\in \fn{v}{C_v}$.
\end{lemma}
\begin{proof}
Since $\ff{u_0} \neq \ff{u_k}$, $C_v$ is a transition chain. Moreover, $C_v$ contains $R(v)\cap \partial P$ by definition. Therefore, Lemma~\ref{lemma:Funnel contains Voronoi region} implies that $R(v)\subset \fn{v}{C_v}$.
Since $v = \ff{x}$, we know that $x\in R(v)$ and hence that $x \in \fn{v}{C_v}$. 
\end{proof}

We now proceed to split a given funnel into $O(|\fn{v}{C_v}|)$ apexed triangles that encode the distance function from $v$. To this end, we compute the shortest-path tree $T_v$ of $v$ in $\fn{v}{C_v}$ in $O(|\fn{v}{C_v}|)$ time~\cite{guibasShortestPathQueries}.
We consider the tree $T_v$ to be rooted at $v$ and assume that for each node $u$ of this tree we have stored the geodesic distance $\g{u}{v}$. 

Start an Eulerian tour from $v$ walking in a clockwise order of the edges. Let Let $w_1$ be the first leaf of $T_v$ found, and let $w_2$ and $w_3$ be the next two vertices visited in the traversal. Two cases arise:

\textbf{Case 1:} $w_1, w_2, w_3$ makes a right turn. We define $s$ as the first point hit by the ray apexed at $w_2$ that shoots in the direction opposite to $w_3$. 

We claim that $w_1$ and $s$ lie on the same edge of the boundary of $\fn{v}{C_v}$. Otherwise, there would be a vertex $u$ visible from $w_2$ inside the wedge with apex $w_2$ spanned by $w_1$ and $w_3$.
Note that the first edge of the path $\p{u}{v}$ is the edge $uw_2$. Therefore, $uw_2$ belongs to the shortest-path $T_v$ contradicting the Eulerian order in which the vertices of this tree are visited as $u$ should be visited before $w_3$. Thus, $s$ and $w_1$ lie on the same edge and $s$ can be computed in $O(1)$ time.

At this point, we construct the apexed triangle $\triangle(w_2, w_1, s)$ apexed at $w_2$ with apex function 
$$g(x) = \left\{ \begin{array}{lll}
 |x w_2| + \g{w_2}{v} && \text{if }x\in \triangle(w_2, w_1, s)\\
-\infty&&\text{otherwise}
\end{array}\right.$$
We modify tree $T_v$ by removing the edge $w_1w_2$ and replacing the edge $w_3w_2$ by the edge $w_3s$; see Figure~\ref{fig:Funnel Cover}.

\textbf{Case 2:} $w_1, w_2, w_3$ makes a left turn and $w_1$ and $w_3$ are adjacent, 
then if $w_1$ and $w_3$ lie on the same edge of $\partial P$, we construct an apexed triangle $\triangle(w_2, w_1, w_3)$ apexed at $w_2$ with apex function 
$$g(x) = \left\{ \begin{array}{lll}
|x w_2| + \g{w_2}{v} && \text{if }x\in \triangle(w_2, w_1, w_3)\\
-\infty&&\text{otherwise}
\end{array}\right.$$
Otherwise, let $s$ be the first point of the boundary of $\fn{v}{C_v}$ hit by the ray shooting from $w_3$ in the direction opposite to $w_2$.

By the same argument as above, we can show that $w_1$ and $s$ lie on the same edge of the boundary of $\fn{v}{C_v}$ (and thus, we can compute $s$ in $O(1)$ time). We construct an apexed triangle $\triangle(w_2, w_1, s)$ apexed at $w_2$ with apex function 
$$g(x) = \left\{ \begin{array}{lll}
|x w_2| + \g{w_2}{v} && \text{if }x\in \triangle(w_2, w_1, s)\\
-\infty&&\text{otherwise}
\end{array}\right..$$
We modify the tree $T_v$ by removing the edge $w_1w_2$ and adding the edge $w_3s$; see Figure~\ref{fig:Funnel Cover} for an illustration.

\begin{figure}[tb]
\centering
\includegraphics{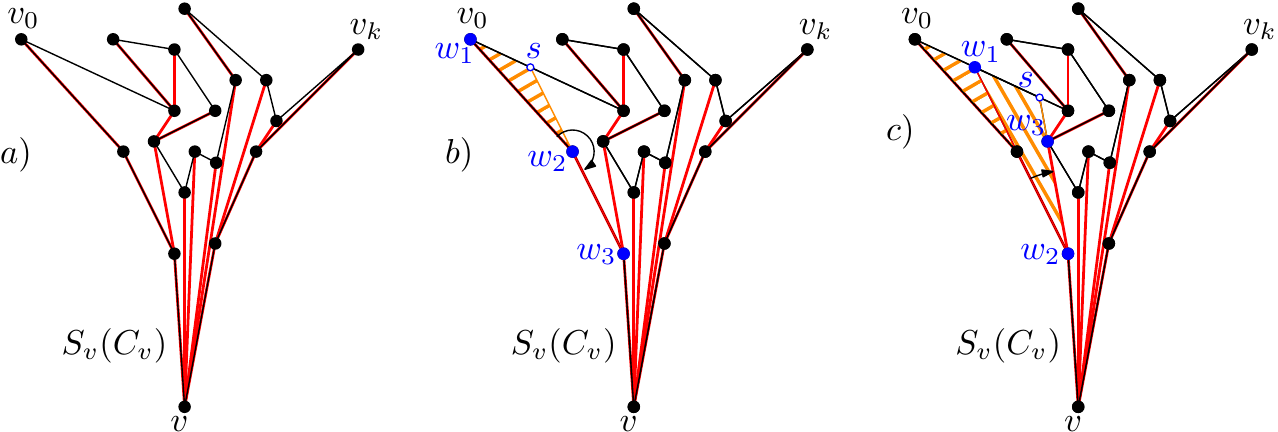}

\caption{\small The funnel $\fn{v}{C_{v}}$ and the shortest-path tree from $v$  are depicted in $(a)$ .
The two cases of the algorithm described in Lemma~\ref{lemma:Triangles inside funnels} are shown in $(b)$ and $(c)$.}
\label{fig:Funnel Cover}
\end{figure}

\begin{lemma}\label{lemma:Triangles inside funnels}
The above procedure runs in $O(|\fn{v}{C_v}|)$ time and computes $O(|\fn{v}{C_v}|)$ interior disjoint apexed triangles such that their union covers $\fn{v}{C_v}$. 
Moreover, for each point $x\in R(v)$,  
there is an apexed triangle $\triangle$ with apex function $g(x)$ such that 
(1) $x\in \triangle$ and (2) $g(x) = \F{P}{x}$.
\end{lemma}
\begin{proof}
The above procedure splits $\fn{v}{C_v}$ into apexed triangles, such that their apex function in each of them is defined as the geodesic distance to $v$. 
By Lemma~\ref{lemma:Farthest points from marked are in funnel}, if $x\in R(v)$, then $x\in \fn{v}{C_v}$. 
Therefore, there is an apexed triangle $\triangle$ with apex function $g(x)$ such that $x\in \triangle$ and $g(x) = \g{x}{v} = \F{P}{x}$. Consequently, we obtain properties (1) and (2).

We now bound the running time of the algorithm. 
The shortest-path tree $T_v$ from $v$ is computed in $O(|\fn{v}{C_v}|)$ time~\cite{guibasShortestPathTree}.
 For each leaf of $T_v$ we need a constant number of operations to determine in which of the cases we are in (and to treat it as well). Therefore, it suffices to bound the number of times these steps are performed.
Note that a leaf is removed from the tree in each iteration. 
Since the number of leaves strictly decreases each time we are in Case 2,  this step cannot happen more than $O(|\fn{v}{C_v}|)$ times. 
In Case 1 a new leaf is added if $w_1$ and $w_3$ do not lie on the same edge of $\partial P$. 
However, the number of leaves that can be added throughout is at most the number of edges of $T_v$. 
Note that the edges added by either Case 1 or 2 are chords of the polygon and hence do not generate further leaves. 
Because $|T_v| = O(|\fn{v}{C_v}|)$, we conclude that both Case 1 and 2 are only executed $O(|\fn{v}{C_v}|)$ times.
\end{proof}

\section{Prune and search}\label{section:Prune and search}
With the tools introduced in the previous sections, we can proceed to give the prune and search algorithm to compute the geodesic center. 
The idea of the algorithm is to partition $P$ into $O(1)$ cells, determine on which cell of $P$ the center lies and recurse on that cell as a new subproblem with smaller complexity.

Naturally, we can discard all apexed triangles that do not intersect the new cell containing the center. Using the properties of the cutting, we can show that both the complexity of the cell containing the center, and the number of apexed triangles that intersect it decrease by a constant fraction in each iteration of the algorithm. This process is then repeated until either of the two objects has constant descriptive size. 

Let $\tau$ be the set all apexed triangles computed in previous sections. Lemmas~\ref{lemma:Bounding complexity of hourglasses} and~\ref{lemma:Triangles inside funnels} directly provide a bound on the complexity of $\tau$.

\begin{corollary}\label{lemma:Size of tau}
The set $\tau$ consists of $O(n)$ apexed triangles.
\end{corollary}

Let $\phi(x)$ be the upper envelope of the apex functions of every triangle in $\tau$ (i.e., $\phi(x) = \max\{g_i(x) : g_i(x)\in\tau\}$). The following result is a direct consequence of Lemmas~\ref{lemma:Triangles inside hourglasses} and \ref{lemma:Triangles inside funnels}, and shows that the $O(n)$ apexed triangles of $\tau$ not only cover $P$, but their apex functions suffice to reconstruct the function $\F{P}{x}$.

\begin{lemma}\label{lemma:Optimization problem same as geodesic center}
The functions $\phi(x)$ and $\F{P}{x}$ coincide in the domain of points of $P$, i.e., for each $p\in P$, $\phi(p) = \F{P}{p}$.
\end{lemma}



Given a chord $C$ of $P$ a \emph{half-polygon} of $P$ is one of the two simple polygons in which $C$ splits $P$.
A \emph{\tcell} of $P$ is a simple polygon obtained as the intersection of at most four half-polygons.
Because a \tcell is the intersection of geodesically convex sets, it is also geodesically convex.


Let $R$ be a \tcell of $P$ and let $\tau_R$ be the set of apexed triangles of $\tau$ that intersect $R$. 
Let $\m = \max\{|R|, |\tau_R|\}$.
Recall that, by construction of the apexed triangles, for each triangle of $\tau_R$ at least one and at most two of its boundary segments is a chord of $P$
Let $\mathcal C$ be the set containing all chords that belong to the boundary of a triangle of $\tau_R$. 
Therefore, $|\tau_R| \leq |\mathcal C| \leq 2|\tau_R|$.

To construct an $\varepsilon$-net of $\mathcal C$, we need some definitions (for more information on $\varepsilon$-nets refer to~\cite{ConstructionEpsilonNets}).
Let $\varphi$ be the set of all open \tcells of $P$.
For each $t\in \varphi$, let $\mathcal C_t = \{C\in \mathcal C: C\cap t \neq \emptyset\}$ be the set of chords of $\mathcal C$ induced by $t$. 
Finally, let $\varphi_\mathcal C = \{\mathcal C_t : t\in \varphi\}$ be the family of subsets of $\mathcal C$ induced by $\varphi$.

Let $\varepsilon >0$ (the exact value of $\varepsilon$ will be specified later).
Consider the range space $(\mathcal C, \varphi_\mathcal C)$ defined by $\mathcal C$ and $\varphi_\mathcal C$. 
Because the VC-dimension of this range space is finite, we can compute an $\varepsilon$-net $N$ of $(\mathcal C, \varphi_\mathcal C)$ in $O(n/\varepsilon)=O(n)$ time~\cite{ConstructionEpsilonNets}. 
The size of $N$ is $O(\frac{1}{\varepsilon} \log \frac{1}{\varepsilon}) = O(1)$ and its main property is that any \tcell that does not intersect a chord of $N$ will intersect at most $\varepsilon |\mathcal C|$ chords of $\mathcal C$. 

Observe that $N$ partitions $R$ into $O(1)$ sub-polygons (not necessarily \tcells). We further refine this partition by performing a \tcell decomposition. That is, we shoot vertical rays up and down from each endpoint of $N$, and from the intersection point of  any two segments of $N$, see Figure~\ref{fig:Cutting of Chords}. Overall, this partitions $R$ into $O(1)$ \tcells such that each  either $(i)$  is a convex polygon contained in $P$ of at most four vertices, or otherwise $(ii)$ contains some chain of $\partial P$. 
Since $|N| = O(1)$, the whole decomposition can be computed in $O(\m)$ time (the intersections between segments of $N$ are done in constant time, and for the ray shooting operations we walk along the boundary of $R$ once).

\begin{figure}[tb]
\centering
\includegraphics{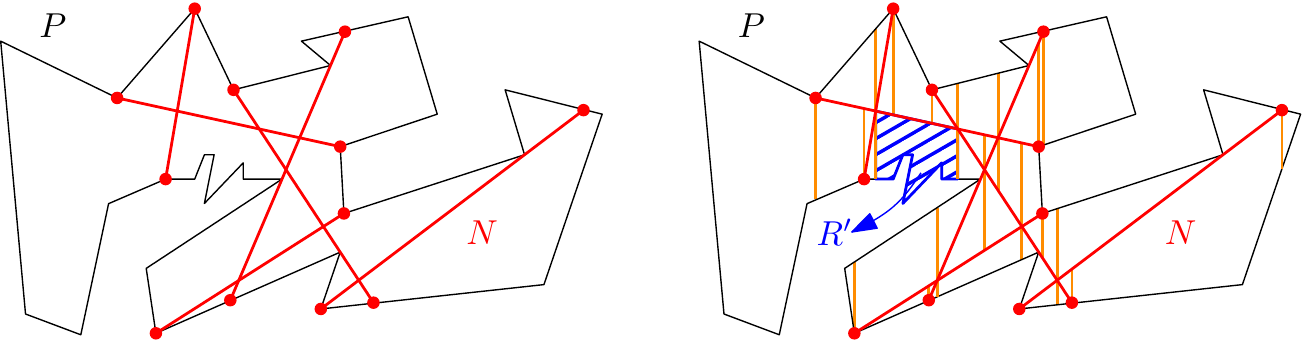}

\caption{\small The $\epsilon$-net $N$ splits $P$ into $O(1)$ sub-polygons that are further refined into a \tcell decomposition using $O(1)$ ray-shooting queries from the vertices of the arrangement defined by~$N$.}
\label{fig:Cutting of Chords}
\end{figure}

In order to determine which \tcell contains the geodesic center of $P$, 
we extend each edge of a \tcell to a chord $C$. 
This can be done with two ray-shooting queries (each of which takes $O(\m)$ time).
We then use the chord-oracle from Pollack et al.~\cite[Section~3]{pollackComputingCenter} to decide which side of $C$ contains $\cp$.
The only requirement of this technique is that the function $\F{P}{x}$ coincides with the upper envelope of the apex functions when restricted to $C$.
Which is true by Lemma~\ref{lemma:Optimization problem same as geodesic center} and from the fact that $\tau_R$ consists of all the apexed triangles of $\tau$ that intersect~$R$.

Because the chord-oracle described by Pollack et al.~\cite[Section~3]{pollackComputingCenter} runs in linear time on the number of functions defined on $C$, we can decide in total $O(\m)$ time on which side of $C$ the geodesic center of $P$ lies. 
Since our decomposition into \tcells has constant complexity, 
we need to perform $O(1)$ calls to the oracle before determining the \tcell $\reg$ that contains the geodesic center of $P$. 

The chord-oracle computes the minimum of $\F{P}{x}$ restricted to the chord before determining the side containing the minimum. In particular, if $\cp$ lies on any chord bounding $\reg$, then the chord-oracle will find it. 
Therefore, we can assume that $\cp$ lies in the interior of $\reg$. Moreover, since $N$ is a $\varepsilon$-net, we know that at most $\varepsilon |\mathcal C|$ chords of $\mathcal C$ will intersect $\reg$.

Using a similar argument, we can show that the complexity of $\reg$ also decreases: since $|\mathcal C| \leq 2|\tau_R| \leq 2\m$, we guarantee that at most $2\varepsilon \m$ apexed triangles intersect~$\reg$. Moreover, each vertex of $\reg$ is in at least one apexed triangle of $\tau_R$ by Lemma~\ref{lemma:Optimization problem same as geodesic center}, and by construction, each apexed triangle can cover at most three vertices. Thus, by the pigeonhole principle we conclude that $\reg$ can have at most $6 \varepsilon \m$ vertices. 
Thus, if we choose $\varepsilon = 1/12$, we guarantee that both the size of the \tcell $\reg$ and the number of apexed triangles in $\tau_{\reg}$ are at most $\m/2$. 


In order to proceed with the algorithm on  $\reg$ recursively, we need to compute the set $\tau_{\reg}$ with the at most $\varepsilon |\mathcal C|$ apexed triangles of $\tau_R$ that intersect $\reg$ (i.e., prune the apexed triangles that do not intersect with $\reg$). For each apexed triangle $\triangle\in \tau_R$, we can determine in constant time if it intersects $\reg$ (either one of the endpoints is in $\reg\cap \partial P$ or the two boundaries have non-empty intersection in the interior of $P$). 
Overall, we need $O(\m)$ time to compute the at most $\varepsilon |\mathcal C|$ triangles of $\tau_R$ that intersect $\reg$.

By recursing on $\reg$, we guarantee that after $O(\log \m)$ iterations, we reduce the size of either $\tau_R$ or $\reg$ to constant. 
In the former case, the minimum of $\F{P}{x}$ can be found by explicitly constructing function $\phi$ in $O(1)$ time. 
In the latter case, we triangulate $\reg$ and apply the chord-oracle to determine which triangle will contain $\cp$. 
The details needed to find the minimum of $\phi(x)$ inside this triangle are giving the next section.


\begin{lemma}\label{lemma:Finding the convex trapezoid}
In $O(n)$ time we can find either the geodesic center of $P$ or a triangle containing the geodesic center.
\end{lemma}

\section{Solving the problem restricted to a triangle}\label{Section:Solving convex optimization poblem}
In order to complete the algorithm it remains to show how to find the geodesic center of $P$ for the case in which $\reg$ is a triangle. If this triangle is in the interior of $P$, it may happen that several apexed triangles of $\tau$ fully contain $\reg$. 
Thus, the pruning technique used in the previous section cannot be further applied. We solve this case with a different approach.
  
Recall that $\phi(x)$ denotes the upper envelope of the apex functions of the triangles in~$\tau$, and the geodesic center is the point that minimizes $\phi$.
The key observation is that, as it happened with chords, the function $\phi(x)$ restricted to $\reg$ is convex. 

Let $\triangle_{1}, \triangle_{2}, \ldots, \triangle_{m}$ be the set of $m= O(n)$ apexed triangles of $\tau$ that intersect $\reg$. 
Let $g_i(x)$  be the apex function of $\triangle_i$ such that 
$$g(x) = \left\{ \begin{array}{lll}
|x a_i| + \kappa_i && \text{if }x\in \triangle_i\\
-\infty&&\text{otherwise}
\end{array}\right. ,$$
where $a_i$ and $w_i$ are the apex and the definer of $\triangle_i$, respectively, and $\kappa_i = \g{a_i}{w_i}$ is a constant.

By Lemma~\ref{lemma:Optimization problem same as geodesic center}, $\phi(x) = \F{P}{x}$. 
Therefore, the problem of finding the center is equivalent to the following optimization problem in $\mathbb{R}^3$:

\textbf{(P1).} Find a point $(x,r)\in \mathbb{R}^3$ minimizing $r$ subject to $x\in \reg$ and
$$\text{$g_i(x) \leq r$, for $1\leq i \leq m$}.$$

Thus, we need only to find the solution to (P1) to find the geodesic center of $P$.
We use some remarks described by Megiddo in order to simplify the description of (P1)~\cite{megiddo1989ball}.

To simplify the formulas, we square the equation $|x a_i| \leq r - \kappa_i$:
$$\|x\|^2 - 2x\cdot a_i + \|a_i\|^2  = |x a_i|^2 \leq (r - \kappa_i)^2 = r^2 - 2r\kappa_i + \kappa_i^2.$$ 
And finally for each $1\leq i\leq m$, we define the function $h_i(x, r)$ as follows:
$$h_i(x, r) = \left\{ \begin{array}{lll}
 \|x\|^2 - 2x\cdot a_i + \|a_i\|^2  - r^2 + 2r\kappa_i - \kappa_i^2 && \text{if }x\in \triangle_i\\
-\infty&&\text{otherwise}
\end{array}\right. $$

Therefore, our optimization problem can be reformulated as:

\textbf{(P2).} Find a point $(x,r)\in \mathbb{R}^3$ such that $r$ is minimized subject to $x\in \reg$ and 
$$h_i(x, r) \leq 0 \text{ and  $r > \max\{\kappa_i\}$, for $1\leq i \leq m$}.$$

Let $h_i'(x,r) = \|x\|^2 - 2x\cdot a_i + \|a_i\|^2  - r^2 + 2r\kappa_i - \kappa_i^2$ be a function defined in the entire plane and let (P2$'$) be an optimization problem analogous to (P2) where every instance of $h_i(x,r)$ is replaced by $h_i'(x,r)$.
The optimization (P2$'$) was studied by Megiddo in~\cite{megiddo1989ball}. We provide some of the intuition used by Megiddo to solve this problem.

Although the functions $h_i'(x,r)$ are not linear inside $\triangle_i$, they all have the same non-linear terms. Therefore, for $i\neq j$, we get that
$h_i'(x,r) = h_j'(x, r)$ defines a \emph{separating plane}
$$\gamma_{i,j} = \{(x, r) \in \mathbb{R}^3:  2( \kappa_i - \kappa_j) r - 2 (a_i - a_j) \cdot x + \|a_i\|^2 - \|a_j\|^2 - \kappa_i^2 + \kappa_j^2 = 0\}$$

As noted by Megiddo~\cite{megiddo1989ball}, this separating plane has the following property:
If the solution $(x, r)$ to (P2$'$) is known to lie to one side of $\gamma_{i,j}$, then we know that one of the constraints is redundant. 

Thus, to solve (P2$'$) it sufficed to have a \emph{side-decision oracle} to determine on which side of a plane $\gamma_{i,j}$ the solution lies. Megiddo showed how to implement this oracle in a way that the running time is proportional to the number of constraints~\cite{megiddo1989ball}.

Once we have such an oracle, Megiddo's problem can be solved using a prune and search similar to that introduced in Section~\ref{section:Prune and search}: pair the functions arbitrarily, and consider the set of $m/2$ separating planes defined by these pairs. For some constant $r$, compute a $1/r$-cutting in $\mathbb{R}^3$ of the separating planes.
A $1/r$-cutting is a partition of the plane into $O(r^2)$ convex regions each of which is of constant size and intersects at most $m/2r$ separating planes.
A cutting of planes can be computed in linear time in $\mathbb{R}^3$ for any $r = O(1)$~\cite{matousekCuttings}.
After computing the cutting, determine in which of the regions the minimum lies by performing $O(1)$ calls to the side-decision oracle. 
Because at least $(r-1)m/2r$ separating planes do not intersect this constant size region, for each of them we can discard one of the constraints as it becomes redundant. Repeating this algorithm recursively we obtain a linear running time.

To solve (P2) we follow a similar approach, but our set of separating planes needs to be extended in order to handle apex functions as they are only defined in the same way as in (P2$'$) in a triangular domain.
Note that the vertices of each apexed triangle that intersect $\reg$ have their endpoints either outside of $\reg$ or on its boundary.

\subsection{Optimization problem in a convex domain}
In this section we describe our algorithm to solve the optimization problem (P2). 
To this end, we pair the apexed triangles arbitrarily to obtain $m/2$ pairs.
By identifying the plane where $P$ lies with the plane $Z_0 = \{(x,y,z): z = 0\}$, we can embed each apexed triangle in $\mathbb{R}^3$.
A \emph{plane-set} is a set consisting of at most five planes in $\mathbb{R}^3$.
For each pair of apexed triangles $(\triangle_i, \triangle_j)$ we define a plane-set as follows: 
For each chord bounding either $\triangle_i$ or $\triangle_j$, consider the line extending this chord and the vertical extrusion of this line in $\mathbb{R}^3$, i.e.,  the plane containing this chord orthogonal to $Z_0$. Moreover, consider the separating plane~$\gamma_{i,j}$. The set containing these planes is the plane-set of the pair $(\triangle_i, \triangle_j)$.

Let $\Gamma$ be the union of all the plane-sets defined by the $m/2$ pairs of apexed triangles. Thus, $\Gamma$ is a set that consists of $O(m)$ planes. Compute an $1/r$-cutting of $\Gamma$ in $O(m)$ time for some constant $r$ to be specified later.
Because $r$ is constant, this $1/r$-cutting splits the space into $O(1)$ convex regions, each bounded by a constant number of planes~\cite{matousekCuttings}. 
Using a side-decision algorithm (to be specified later), we can determine the region $Q$ of the cutting that contains the solution to (P2). Because $Q$ is the region of a $1/r$-cutting of $\Gamma$, we know that at most $|\Gamma|/r$ planes of $\Gamma$ intersect $Q$. In particular, at most $|\Gamma|/r$ plane-sets intersect $Q$ and hence, at least $(r-1)|\Gamma|/r$ plane-sets do not intersect $Q$. 

Let $(\triangle_i, \triangle_j)$ be a pair such that its plane-set does not intersect $Q$. 
Let $Q'$ be the projection of $Q$ on the plane $Z_0$. Because the plane-set of this pair does not intersect $Q$, we know that $Q'$ intersects neither the boundary of $\triangle_i$ nor that of $\triangle_j$.
Two cases arise:

\textbf{Case 1.} If either $\triangle_i$ or $\triangle_j$ does not intersect $Q'$, then we know that their apex function is redundant and we can drop the constraint associated with this apexed triangle.

\textbf{Case 2.} If $Q'\subset \triangle_i\cap \triangle_j$, then we need to decide which constrain to drop. 
To this end, we consider the separating plane $\gamma_{i,j}$. Notice that inside the vertical extrusion of $\triangle_i\cap \triangle_j$ (and hence in $Q$), the plane $\gamma_{i,j}$ has the property that if we know its side containing the solution, then one of the constraints can be dropped. Since $\gamma_{i,j}$ does not intersect $Q$ as $\gamma_{i,j}$ belongs to the plane-set of $(\triangle_i, \triangle_j)$, we can decide which side of $\gamma_{i,j}$ contains the solution to (P2) and drop one of the constraints.
\vspace{.05in}

Regardless of the case if the plane-set of a pair $(\triangle_i, \triangle_j)$ does not intersect $Q$, then we can drop one of its constraints. Since at least $(r-1)|\Gamma|/r$ plane-sets do not intersect~$Q$, we can drop at least $(r-1)|\Gamma|/r$ constraints.
Because $|\Gamma| \geq m/2$ as each plane-set contains at least one plane, by choosing $r = 2$, we are able to drop at least $|\Gamma|/2 \geq m/4$ constraints.
Consequently, after $O(m)$ time, we are able to drop $m/4$ apexed triangles.
By repeating this process recursively, we end up with a constant size problem in which we can compute the upper envelope of the functions explicitly and find the solution to (P2) using exhaustive search. 
Thus, the running time of this algorithm is bounded by the recurrence $T(m) = T(3m/4) + O(m)$ which solves to $O(m)$. 
Because $m = O(n)$, we can find the solution to (P2) in $O(n)$ time.

The last detail is the implementation of the side-decision algorithm. 
Given a plane~$\gamma$, we want to decide on which side lies the solution to (P2).
To this end, we solve (P2) restricted to~$\gamma$, i.e., with the additional constraint of $(x,r)\in \gamma$. 
This approach was used by Megiddo~\cite{megiddo1989ball}, the idea is to recurse by reducing the dimension of the problem.
Another approach is to use a slight modification of the chord-oracle described by Pollack et al.~\cite[Section~3]{pollackComputingCenter}. 

Once the solution to (P2) restricted to $\gamma$ is known, we can follow the same idea used by Megiddo~\cite{megiddo1989ball} to find the side of $\gamma$ containing the global solution to (P2). 
Intuitively, we find the apex functions that define the minimum restricted to $\gamma$. 
Since $\phi(x) = \F{P}{x}$ is locally defined by this functions, 
we can decide on which side the minimum lies using convexity.
We obtain the following result.

\begin{lemma}
Let $\reg$ be a convex trapezoid contained in $P$ such that $\reg$ contains the geodesic center of $P$. 
Given the set of all apexed triangles of $\tau$ that intersect $\reg$, 
we can compute the geodesic center of $P$ in $O(n)$ time.
\end{lemma}

The following theorem summarizes the result presented in this paper.

\begin{theorem}
We can compute the geodesic center of any simple polygon $P$ of $n$ vertices in $O(n)$ time.
\end{theorem}

\bibliographystyle{abbrv}
\bibliography{Geodesic}

\end{document}